\documentclass[journal]{IEEEtran}

\usepackage{amsmath}
\usepackage{amsthm}
\usepackage{cases}
\usepackage{amsfonts}
\usepackage{amssymb}
\usepackage{boxedminipage}
\usepackage{graphicx}
\usepackage[usenames]{color}
\usepackage{array,color}
\usepackage{colortbl}
\usepackage{extarrows}
\usepackage{bm}
\usepackage{amsmath}
\usepackage[numbers,sort&compress]{natbib}

\newtheorem{theorem}{Theorem}
\newtheorem{lemma}{Lemma}

\newtheorem{proposition}{Proposition}
\newtheorem{remark}{Remark}

\begin{document}
\baselineskip=12pt

\title{ Mobility-Aware Uplink Interference Model for 5G Heterogeneous Networks }

\author{Yunquan~Dong,~\IEEEmembership{Member,~IEEE},
        Zhi Chen,~\IEEEmembership{Member,~IEEE},\\
        Pingyi~Fan,~\IEEEmembership{Senior Member,~IEEE},
        and Khaled~Ben~Letaief,~\IEEEmembership{Fellow,~IEEE}
        \thanks{
            Y. Dong is with the Department of Electrical and Computer Engineering, Seoul National University, Seoul, Korea 151744.
            Y. Dong was with the Department of Electrical Engineering, Tsinghua University,
            Beijing, China, 100084. Email: ydong@snu.ac.kr.

            Z. Chen is with the Department of Electrical and Computer Engineering, University of Waterloo, Waterloo, Ontario, Canada, N2L3G1, Email: z335chen@uwaterloo.ca.

            P. Fan is with the Department of Electrical Engineering, Tsinghua University, Beijing, China, 100084. Email: fpy@tsinghua.edu.cn.

            Khaled B. Letaief is with the Hamad bin Khalifa University, Qatar (kletaief@hbku.edu.qa). He is also with the Department of Electrical and Computer Engineering, HKUST, Clear Water Bay, Kowloon, Hong Kong (eekhaled@ust.hk)
            }
        }

\maketitle

\begin{abstract}
  To meet the surging demand for throughput, 5G cellular networks need to be more heterogeneous and much denser, by deploying more and more small cells.
   In particular, the number of users in each small cell can change dramatically due to users' mobility, resulting in random and time varying uplink interference.
  This paper considers the uplink interference in a 5G heterogeneous network which is jointly covered by one macro cell and several small cells.
   Based on the L\'{e}vy flight moving model, a mobility-aware interference model is proposed to characterize the uplink interference from macro cell users to small cell users.
  In this model, the total uplink interference is characterized by its moment generating function, for both closed subscriber group (CSG) and open subscriber group (CSG) femto cells.
   In addition, the proposed interference model is a function of basic step length, which is a key velocity parameter of L\'{e}vy flights.
  It is shown by both theoretical analysis and simulation results that the proposed interference model provides a flexible way of evaluating the system performance in terms of success probability and average rate.
\end{abstract}


\section{Introduction}
As the long term evolution/advanced (LTE/LTE-A) cellular system has been deployed all over the world and is reaching maturity,  the standards bodies and industry are now organizing a timeframe to standardize the fifth generation (5G) technology, which is expected to be between 2016 and 2018, followed by initial deployments around 2020.
As is expected, the network aggregate data rate will be increased by roughly 1000x from 4G to 5G \cite{Andrews-2014, Boccardi-2014}.
To achieve this ambition, 5G communication systems need more nodes per unit area besides more Hz and more bit/s/Hz per node \cite{Andrews-2014}.
Therefore, more and more small cells such as pico/femto/relay cells are being added to the existing network \cite{Andrews-2013}.
In this context, it may not be surprising to expect that in the not too distant future, the number of base stations may exceed the number of cell phone subscribers \cite{Andrews-2013}. A network that consists of a mix of macro cells and small cells is often referred to as a {\em heterogeneous network} (HetNet), or \textit{DenseNets} \cite{Ahmad-2013}.

Research on HetNets dates back to the discussion on femto cells in 2008 \cite{Claussen-2008} and was admitted by 3GPP LTE-A standard in 2011 \cite{3Gpp-2011}. HetNets is also believed to be an important part of the next generation cellular networks. By adding more and more low power small cells, the reuse of spectrum across the space is improved.
At the same time, the number of users competing for resources at each base station is reduced.
Note that, the spectral efficiency of modern access technologies such as LTE is
already very close to the Shannon's limit \cite{Ahmad-2013}.
Therefore, enhancing the network efficiency by densifying the network in the spatial domain rather than user efficiency in the frequency domain would be one important step towards 5G communications.

Due to the scarcity of spectrum, lower power base stations are preferred to be deployed in the same band as macro base stations. Naturally, the interference management in HetNets becomes an unavoidable issue. In the literature, this problem has been discussed from various viewpoints. In the physical layer, the downlink co-channel interference can be modeled as an interference channel. Based on this observation,  an interference canceling block modulation scheme was proposed in \cite{Ayyar-2012}, in which interference can be canceled successively since the covariance matrix of the interference is designed to be rank deficient at each receiver.
In \cite{Grant-2002, Kuchi-2011}, joint detection algorithms and maximum-likelihood based local detections were proposed. In the MAC layer or above, related issues include: 1) frequency reuse techniques such as fractional frequency reuse \cite{Fujii-2008} or soft frequency reuse \cite{Huawei-2005} and optimum combining \cite{Winters-2002}, 2) load balancing and power control schemes such as range extension technique in 3GPP LTE Rel-10 systems, user association schemes in \cite{andrews-user2013,Andrews-balc-2013}, and the proportional optimal power control in \cite{Li-2013},  and 3) fundamental research on interference modeling  \cite{Health-2013,Tabassum-2013,Iyer-2009} that will facilitate interference management.

Among them, the authors of \cite{Iyer-2009} investigated the difference as well as the equivalence among some commonly used interference models for adhoc/sensor networks, such as the additive interference model, the capture threshold model, the protocol model and the interference range model.
As pointed in \cite{Iyer-2009}, different interference models can produce significantly different results.
The uplink intercell interference modeling for HetNets was investigated in \cite{Tabassum-2013}, in which the distribution of the location of scheduled users and the moment generating function (MGF) of their interference were found. Most recently, \cite{Health-2013} studied the downlink interference in a HetNet using stochastic geometry theory, in which every interfering base station, locating outside of a guard zone, follows the Poisson point process (PPP). A dominant interferer was also assumed to locate at the edge of the guard zone. Together with the Gamma approximation method, the Laplace transform of total interference was given, which can be used to evaluate users' success probability and average rate.

Although very important, previous works focused on static adhoc networks or HetNets, where the interferes are fixed. However, the uplink interference in 5G HetNets are produced by users with mobility.
In addition, among those works considering user mobility in HetNets, most of them were investigating how user mobility affected handover performances \cite{Bennis-handover, lopez-ICIC}. As a result, it is still not clear whether users' mobility will change uplink interference model or not, which is the motivation of this paper.

Particularly, since more and more base stations are deployed in the network, each cell becomes smaller and smaller. As a result, the number of users in a cell or an interfering area is very limited. In this case, users' mobility will have more impact on the number of users in a cell, which determines their uplink interferences to other type of users in the same area.

This paper focuses on characterizing the uplink interference in 5G HetNets based on the  \textit{L\'{e}vy flights} \cite{Rhee-levyflights} moving model.
In this model, each user moves one step in every time interval $T_s$.
Formally, when a user moves from one location to another without a directional change or pause, a flight is defined as the longest straight-line distance between its starting point and end point.
Some recent studies on human mobility show that the flight length distributions have a heavy-tail tendency \cite{Lee-levyflights, Rhee-levyflights}.
By normalizing the flight length with a \textit{basic step length} $\Delta$, the flight length turns to be the number of basic steps in each flight.
Therefore, $\Delta$ can be seen as an indicator of user velocity, i.e., a user with a larger $\Delta$ moves more quickly on average.

Under the L\'{e}vy flight mobility model, the number of interferers in an interfering region varies from time to time.
   Particularly, the number of users in the interfering region can be modeled by a Markov process, in which the state transition probability is a function of user mobility.
It is also seen that users may sometimes move out of the macro cell due to mobility.
   To eliminate this kind of boundary effect, we proposed a \textit{modified reflection model}, which can simulate the network using a single macro cell.
     In this model, a user is assumed to re-enter the macro cell from the opposite edge when it reaches the cell edge, without changing its moving direction.
In addition, this paper considers the uplink interference for both closed subscriber group (CSG) femto cells which serve some authenticated users only, and the open subscriber group (OSG) femto cells which admit any user coming into its coverage, by deriving their moment generating functions, mean values and variations.
   It will be seen from simulation results that the uplink interference is actually a constant, regardless how fast users move.
It is also shown that the proposed interference model is useful in evaluating the system performance such as the probability of successful transmission and the average transmission rate.


The rest of this paper is organized as follows. The system model is presented in Section \ref{sec:2}.
   User's average probability of coming into or gonging out of a small cell is presented in Section \ref{sec:3}, based on which the number of users in a small cell is formulated as a Markov Chain in Section \ref{sec:4}.
After that, the uplink interference is presented in terms of its statistics in Section \ref{sec:5}.
As its two applications, the interference model will be used to evaluate user's success probability
and average transmission rate in Section \ref{sec:6}.
The obtained result will also be presented via numerical and Monte Carlo simulation results in Section \ref{sec:7}. Finally, we will conclude this work in section \ref{sec:8}.

\section{System Model}\label{sec:2}
Fig. \ref{fig:net model} presents one macro cell of a heterogeneous network.
The area is covered by a macro-eNB (M-eNB), as well as some low power pico-eNBs, femto-eNBs and relays (collectively referred to as home-eNBs, H-eNBs) \cite{comag-2012, 3Gpp-220}.
The small cells served by these H-eNBs can either be OSG cells or CSG cells.
Denote the radius of the macro cell as $L$, the radius of a small cell of interest (a pico/femto/relay cell, OSG or CSG) as $R$. In general, $L$ is much larger than $R$.

\begin{figure}[h]
\centering
\includegraphics[width=2.8in]{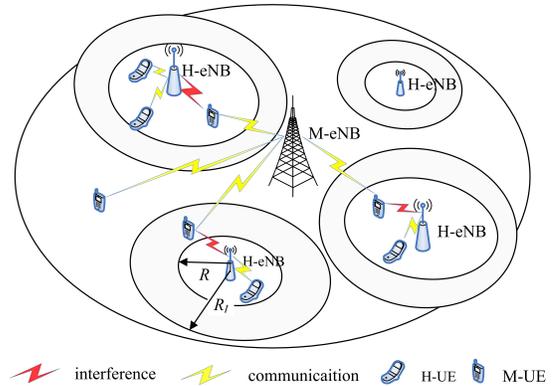}
\caption{5G Heterogeneous Network model.} \label{fig:net model}
\end{figure}

 Users served by H-eNBs and M-eNBs are referred to as the home users (H-UEs) and macro users (M-UEs), respectively.
    Since pico/femto/relay cells are much smaller than macro cells, the transmit power of H-UEs' ($P_t^h$) will be much smaller than that of M-UEs.
As a result, the uplink interference from M-UEs to H-UEs is very strong.
Due to large scale attenuation and small scale fading, the received signal power at an e-NB can be given by
$P_r=\frac{\gamma}{d^\beta}P_t$,
where $d$ is the distance between the user and the e-NB, $\beta\geq2$ is the pathloss exponent, and $\gamma$ is the random channel power gain.
     Without loss of generality, Rayleigh fading model will be used in out simulations.

Due to large scale attenuation, both the desired signal and the interference will be attenuated greatly.
   Similar to the interference range model in \cite{Iyer-2009}, this paper assumes that only the interference from M-UEs within an interfering circle will be considered, as shown by Fig. \ref{fig:net model}.
Denote the radius of the interfering circle as \textit{interfering radius} $R_I$, which is usually larger than the cell radius $R$.
   In addition, it is assumed that users in the same small cell will access to the H-eNB in a time division multiple access (TDMA) manner, so that interference among them is avoided.
Since the transmit power of H-UEs is low and the attenuation is high, interference from other small cells are also neglected.

Assume that there are $N$ users distributed uniformly in the macro cell.
   Assume that all the incoming traffic from the users can be absorbed by the network.
Due to mobility, each user will move to a new location in every time interval $T_s$, according to the \textit{L\'{e}vy flight} model.
   In this model, each move of a user is defined as a \textit{flight}.
The direction of a flight is uniformly distributed among $[0,2\pi)$ and the flight length $X$  follows the power law distribution. Particularly, the probability density function (\textit{pdf}) of $X$ is given by
\begin{equation}\label{eq:powerlawX}
  f_X(x)=\frac{\alpha \Delta^\alpha}{x^{\alpha+1}},\quad x\in[\Delta,+\infty)
\end{equation}
where $\alpha$ falls in between 0.53 and 1.81, as shown by many human mobility traces \cite{Lee-levyflights, Rhee-levyflights}.
By normalizing the flight length using a \textit{basic step length} $\Delta$, one will get $Z=\frac{X}{\Delta}$ with $f_Z(z)=\frac{\alpha}{z^{1+\alpha}},~ z\in[1,+\infty)$.

It is clear that users tend to take longer flights if $\Delta$ is larger.
Therefore, the general moving velocity is determined by the basic step length $\Delta$, for any given $T_s$.
In this sense, this paper will show whether user mobility will affect uplink interference, by investigating the functional relationships between the statistics of uplink interference and  $\Delta$.

To simulate the whole network using a single macro cell, this paper proposed a \textit{modified reflection model}, as shown in Fig. \ref{fig:reflec}.
Assume that a user moves from point $O$ along the direction $\overrightarrow{OQ_1}$.
 Suppose that the flight length is so large that the user tends to leave the macro cell from point $Q_1$. Under the modified reflection model, the user will enter the macro cell again from its opposite point on the cell edge, i.e., point $Q_2$, along the same direction. If the flight is so large that the user can leave the macro cell once again from point $Q_3$, then it will re-enter the macro cell from point $Q_4$, and so on. Under this model, it is noted that the number of users in the macro cell will not change.

\begin{figure}[h]
\centering
\includegraphics[width=2in]{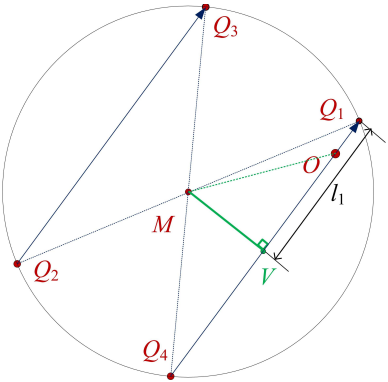}
\caption{Revised reflection model. The user starts from point $O$. $M$ is the center of the macro cell. The user may leave the macro cell from points $Q_1,~Q_3$ and re-enter from points $Q_2,~Q_4$, respectively.}\label{fig:reflec}
\end{figure}

Actually, the number of users in a small cell is a random variable. As a result, the uplink interference is also random.
Let $C_k$ be an arbitrary chosen small cell with radius $R$,  and $\xi_n$ be the number of users in $C_k$ at the beginning of  time interval $[nT_s,(n+1)T_s]$, it can be seen that the process $\{\xi_n,n\geq0\}$ is a Markov chain.
In order to show this, the average probability that a user moves into and goes out of a small cell will be discussed first.

\section{Average Incoming/Outgoing Probability}\label{sec:3}

Since each flight can take any length no shorter than $\Delta$ in any direction, every user outside of a small cell of interest $C_k$ has the chance to come into the cell. Likewise, any user in $C_k$ may move out with some probability.
   For a user who is outside of $C_k$, define the probability that it comes into $C_k$ after one move as its incoming probability.
For a user who lies in $C_k$, define the probability that it goes out of $C_k$ after one move as its outgoing probability.
   By taking average over all possible user locations, the average incoming probability $P_i(R, \Delta)$ and average outgoing probability $P_o(R, \Delta)$ can be obtained, where $R$ is the radius of $C_k$ and $\Delta$ is the basic move length.

Assume that the initial location of each user is uniformly distributed in the macro cell.  It will be shown by the following lemma that, the location of each user is uniformly distributed throughout the operation. In fact, this can be readily understood since each user moves in a pure random way.

\begin{lemma}\label{lem:unif_distri}
   The end point of a random flight will be uniformly distributed in the macro cell, if its start point follows the uniform distribution.
\end{lemma}

\textit{Proof:} See Appendix \ref{prf:lem1}.

\subsection{Average Outgoing Probability}
\begin{figure}[h]
\centering
\includegraphics[width=2in]{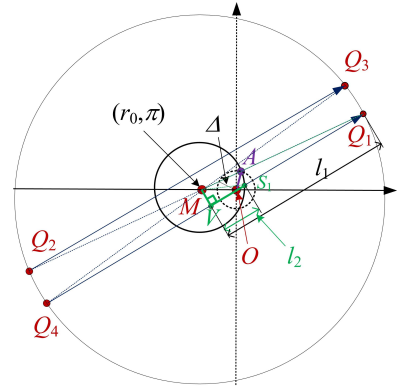}
\caption{The outgoing probability. The user $i$ starts from point $O$, which is also the origin. Both the macro cell and the small cell $C_k$ are centered at point $M$, i.e., point $(r_0, \pi)$. The flight intersect with $C_k$ at point $S_1$, i.e., $(r_\theta, \theta)$,  $l_1=|VQ_1|$, $l_2=|VS_1|$.}\label{fig:pout}
\end{figure}

As shown in Fig. \ref{fig:pout}, a certain user UE$_i$ locates at the origin of the polar coordinate system, i.e., point $O$.
 The center of the macro cell is point $M$,  $(r_0,\pi)$.
There is also a small cell $C_k$ (femto/pico/relay cells, CSG or OSG) centered at point $M$. The radius of cell $C_k$ is $R$.
   Due to the isotropic property of the circular cell, the relative position of a user to $C_k$  depends only on its distance to the cell center.
 Therefore, it is sufficient to consider users at different locations by changing the cell center of $C_k$, i.e., changing $r_0$.

For any given $r_0$, suppose that UE$_i$ moves along direction $\theta$, which will intersect $C_k$ at points  $S_1,~(r_\theta, \theta)$. It is seen that $(r_\theta, \theta)$  satisfies the following equation
\begin{equation}\label{dr:po_rtheta2}
  (r_0+r_\theta \cos \theta)^2+(r_\theta \sin \theta)^2=R^2.
\end{equation}

From (\ref{dr:po_rtheta2}), we have
\begin{equation}\label{dr:po_rtheta3}
  r_\theta=\sqrt{R^2 - r_0^2\sin^2\theta}-r_0\cos\theta.
\end{equation}

Next, the outgoing probability can be solved case by case. In cases 1) to 3), it is assumed that the flight length is relatively small so that UE$_i$ will not move out of the macro cell.  The outgoing probability of going out for large flights is considered in case 4).

\subsubsection{$\Delta<R$}\label{sub_sec_1}
If $0<r_0<R-\Delta$, it is seen that $r_\theta>\Delta$ for every $\theta\in[0,\pi)$. Then UE$_i$ can move out of $C_k$ if the flight length satisfies $X>r_\theta$. Then we have
\begin{equation}\label{rt:po1}\nonumber
  P_o(r_0|0<r_0<R-\Delta)=2\int_{0}^\pi\frac{1}{2\pi} \Pr\{X>r_\theta\} d\theta
\end{equation}
where $P_o(r_0|A)$ represents the outgoing probability as a function of $r_0$ when it is conditioned on event $A$.

If $R-\Delta<r_0<R$, UE$_i$ will be very close to the edge. In this case, UE$_i$ will move out of $C_k$ directly in some scenarios, since every flight is not shorter than $\Delta$.

Define $\theta_1$ as the angle which enables UE$_i$ to reach the cell edge when the flight length is exactly $\Delta$.
Then point $(\Delta,\theta_1)$ lies on the curve defined by (\ref{dr:po_rtheta2}). Thus, $\theta_1$ can be solved as
\begin{equation}\label{dr:theta1}\nonumber
  \theta_1=\arccos{\frac{R^2-r_0^2-\Delta^2}{2\Delta r_0}},\quad \theta_1\in(0,\pi).
\end{equation}

It is clear that $r_\theta<\Delta$ holds true if $\theta\in(-\theta_1,\theta_1)$, which means that UE$_i$ will certainly move out of $C_k$. For any $\theta\notin(-\theta_1,\theta_1)$, the user can move out only if the flight length satisfies $X>r_\theta$. Then the conditional outgoing probability is
\begin{equation}\label{rt:po2}\nonumber
\begin{split}
  P_o(r_0|R-\Delta<r_0<R) \\
        =2\int_0^{\theta_1}\frac{1}{2\pi}d\theta &+2\int_{\theta_1}^\pi\frac{1}{2\pi} \Pr\{X>r_\theta\} d\theta.
\end{split}
\end{equation}

\subsubsection{$R\leq\Delta<2R$} In this situation, UE$_i$ will move out of $C_k$ directly in most directions, except that $\theta \in(\theta_1,2\pi-\theta_1)$ and $r_0>\Delta-R$. We have
\begin{equation}\label{rt:po3}\nonumber
\begin{split}
  &P_o(r_0|0<r_0<\Delta-R)=1,\\
  &P_o(r_0|\Delta-R<r_0<R)\\
   &\qquad~~~~  =2\int_0^{\theta_1}\frac{1}{2\pi}d\theta +2\int_{\theta_1}^\pi\frac{1}{2\pi} \Pr\{X>r_\theta\} d\theta.
\end{split}
\end{equation}

\subsubsection{$\Delta \geq 2R$} \label{sub_sec_3}
In this situation, UE$_i$ will move out of $C_k$ with probability 1, regardless of its moving direction and location, i.e.,
$P_o(r_0|2R<\Delta)=1$.

\subsubsection{The case when flight length is very large}
Finally, It is noted that if the flight length $X$ is very large, it is possible that UE$_i$ will move out of the macro cell from point $Q_1$ and re-enter from point $Q_2$. In fact, UE$_i$ may move out of the macro cell and re-enter for many times. The probability that UE$_i$ will come back to $C_k$ is
\begin{equation}\label{eq:notout}
\begin{split}
  P_o^{re}(r_0)=&\Pr\{~\mathrm{return~to}~ C_k\}\\
        =&\sum_{m=1}^\infty \Pr\{r_\theta + 2ml_1-2l_2<X<r_\theta + 2ml_1\}
\end{split}
\end{equation}
where $l_1=|VQ_1|=\sqrt{L^2-r_0^2\sin^2{\theta}}$ and $l_2=|VS_1|=\sqrt{R^2-r_0^2\sin^2{\theta}}$ are the half chord length within the macro cell and small cell $C_k$, respectively.

By Lemma\ref{lem:unif_distri}, the location of UE$_i$ is uniformly distributed in the macro cell. Thus the probability that $r_0$ is smaller than $x$ is $F_{r_0}(x)=\Pr\{r_0\leq x\}=\frac{x^2}{R^2}$.
Then we can get the \textit{pdf} of $r_0$ as
$f_{r}(x)=\frac{2x}{R^2},~ x\in[0,R]$,
which is independent from the angle.

According to taking average over $r_0$ and following the analysis above, the proposition below summarizes  the average outgoing probability.

\begin{proposition}\label{prop:po}
   The average outgoing probability that a user in $C_k$ will move out of the cell is given by (\ref{rt:pout}), as shown on the top of next page,
\begin{figure*}
\hrulefill
\begin{equation}\label{rt:pout}
  P_o(R,\Delta)=\left\{ \begin{aligned}
             &\frac{2\Delta^\alpha}{\pi R^2} \int_0^{R-\Delta}\int_0^\pi \frac{r_0}{r_\theta^\alpha} d\theta dr_0
+\frac{2}{\pi R^2} \int_{R-\Delta}^R r_0\left(\theta_1+ \int_{\theta_1}^\pi \frac{\Delta^\alpha}{r_\theta^\alpha} d\theta \right)dr_0- P_o^{re}, & &\Delta<R;\\
             &\frac{(\Delta-R)^2}{R^2}+ \frac{2}{\pi R^2} \int_{\Delta-R}^R r_0\left(\theta_1+ \int_{\theta_1}^\pi \frac{\Delta^\alpha}{r_\theta^\alpha} d\theta \right)dr_0- P_o^{re},        & R\leq&\Delta<2R;\\
             &1- P_o^{re},            &~&\Delta\geq 2R.
             \end{aligned} \right.
\end{equation}
\hrulefill
\end{figure*}
where $\theta_1=\arccos{\frac{R^2-r_0^2-\Delta^2}{2\Delta r_0}}$, $l_1=\sqrt{L^2-r_0^2\sin^2\theta}$, $l_2=\sqrt{R^2-r_0^2\sin^2\theta}$ and
 $P_o^{re}=\frac{2\Delta^\alpha}{\pi R^2} \sum_{m=1}^\infty \int_0^R\int_0^\pi \left( \frac{1}{(r_\theta+2ml_1-2l_2)^\alpha} - \frac{1}{(r_\theta+2ml_1)^\alpha} \right) d\theta dr_0$, .
\end{proposition}

\subsection{Average Incoming Probability}
As shown in Fig. \ref{fig:pin}, UE$_i$ locates at the origin (point $O$) of the polar coordinate system. Its distance to the center of the macro cell (point $M$, also the center of a chosen small cell ($C_k$) is $d_0$.  To evaluate the probability that UE$_i$ comes into $C_k$ when UE$_i$ locates at different locations, it is equivalent to fix the position of UE$_i$ while changing the position of the cell center (point $M$). In addition, since their relative position depends only on the distance between them, only $d_0$ needs to be changed.

\begin{figure}[h]
\centering
\includegraphics[width=2in]{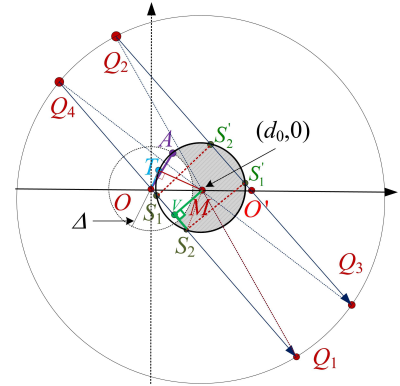}
\caption{The incoming probability. UE$_i$ starts from point $O$, which is the origin. Both the macro cell and the small cell $C_k$ are centered at point $M$, i.e., point $(d_0, 0)$. $\rho_1=|OS_1|,~\rho_2=|OS_2|$, $l_1=|VQ_1|$, $\theta_2=\angle AOM,~\theta_3=\angle TOM$.}\label{fig:pin}
\end{figure}

Any point $(\rho_\theta,\theta)$ locates on the edge of $C_k$ must satisfy
\begin{equation}\label{dr:rtheta_in2}
  (d_0-\rho_\theta\cos\theta)^2+(\rho_\theta\sin\theta)^2=R^2.
\end{equation}

Solving $\rho_\theta$ from this equation, we have two roots of $\rho$
\begin{equation}\label{rt:rou1rou2}\nonumber
  \begin{split}
     \rho_1&=d_0\cos\theta-\sqrt{R^2 - d_0^2\sin^2\theta} \textrm{~~~and}\\
     \rho_2&=d_0\cos\theta+\sqrt{R^2 - d_0^2\sin^2\theta}
  \end{split}
\end{equation}
which corresponds to segment $|OS_1|$ and $|OS_2|$ in Fig. \ref{fig:pin}, respectively.

Define $\theta_2$ as the angular coordinate of the intersection point $A$, i.e., $\angle AOM$. It is seen that  $(\Delta,\theta_2)$ satisfies equation (\ref{dr:rtheta_in2}). Then we have $\theta_2=\arccos{\frac{\Delta^2+d_0^2-R^2}{2\Delta d_0}}$ $\theta_2$ by solving the equation.

Define $\theta_3$ as the angular coordinate of the tangent line of circle $C_k$ which passes the origin, i.e., $\theta_3=\angle TOM$. We have
$\theta_3=\arcsin\frac{R}{d_0}$.

It is noted that $\theta_3\geq\theta_2$ holds true for any $d_0\in[R,L]$. Particularly, solving $d_0$ from $\theta_2=\theta_3$, i.e.,
\begin{equation}\label{dr:d0*}\nonumber
  \arccos{\frac{\Delta^2+d_0^2-R^2}{2\Delta d_0}}=\arcsin\frac{R}{d_0}
\end{equation}
we know that $d_0^*=\sqrt{\Delta^2+R^2}$ satisfies $\theta_2=\theta_3$.

As shown in Fig. \ref{fig:pin}, the incoming probability is also the probability that the end point of a flight falls into  $C_k$. In the following part, we will discuss the incoming probability case by case.

First, assume that the flight length is relatively small and UE$_i$ will come into $C_k$ directly.

\subsubsection{$\Delta<2R$}
In this case, UE$_i$ has non-zero probability to enter $C_k$ if $d_0>R$.

First, if $R<d_0<d_0^*$, it is seen that $\rho_1<\Delta<\rho_2$ if $\theta\in(-\theta_2,\theta_2)$ and $\rho_1<\rho_2<\Delta$ if $\theta\in(-\theta_3, -\theta_2)\cup(\theta_2,\theta_3)$. Then UE$_i$ will be in $C_k$  if the direction of the flight satisfies $\theta\in(-\theta_2,\theta_2)$ and the flight length satisfies $\Delta<X<\rho_2$.
The corresponding conditional incoming probability is
\begin{equation}\label{rt:pin11}\nonumber
\begin{split}
  P_i( d_0| R<d_0<\sqrt{\Delta^2+R^2})&\\
         =2\int_0^{\theta_2} &\frac{1}{2\pi} \Pr\{ \Delta<X<\rho_2 \} d\theta .
\end{split}
\end{equation}

Second, if $d_0^*\leq d_0<R+\Delta$, it is seen that $\rho_1<\Delta<\rho_2$ holds true if $\theta\in(-\theta_2,\theta_2)$. It is also seen that $\Delta<\rho_1<\rho_2$ holds true if $\theta\in(-\theta_3, -\theta_2)\cup(\theta_2,\theta_3)$.  Therefore, UE$_i$ will come into $C_k$ if $\theta\in(-\theta_3,-\theta_2)\cup (\theta_2,\theta_3)$ and the flight length satisfies $\rho_1<X<\rho_2$,  or $\theta\in(-\theta_2,\theta_2)$ and the flight length satisfies $\Delta<X<\rho_2$. The corresponding probability is,
\begin{equation}\label{rt:pin11}\nonumber
\begin{split}
P_i( d_0| \sqrt{\Delta^2+R^2}<d_0<R+\Delta)&\\
            =2\int_0^{\theta_2} \frac{1}{2\pi}\Pr\{ \Delta  & <X<\rho_2 \} d\theta \\
             + 2\int_{\theta_2}^{\theta_3} \frac{1}{2\pi}\Pr\{ \rho_1 &  <X<\rho_2 \} d\theta.
\end{split}
\end{equation}

Third, if $d_0\geq R+\Delta$,  then $\Delta<\rho_1<\rho_2$ holds true for any $\theta\in(-\theta_3,\theta_3)$ and the conditional incoming probability is
\begin{equation}\label{rt:pin11}\nonumber
  P_i( d_0|R+\Delta \leq d_0)= 2\int_{0}^{\theta_3} \frac{1}{2\pi} \Pr\{ \rho_1<X<\rho_2 \} d\theta.
\end{equation}

\subsubsection{$\Delta\geq 2R$}
In this case, UE$_i$ will certainly move across $C_k$ unless $d_0$ is larger than $\Delta-R$. Likewise, the incoming probability can also be obtained, by replacing the lower limit of integrals on $d_0$ with $\Delta-R$.

In addition to what discussed above, UE$_i$ may also come into $C_k$ indirectly. For example, UE$_i$ starts from point $O$ along direction $\overrightarrow{OQ_1}$. If the flight is very large, UE$_i$ may leave the macro cell from point $Q_1$ and re-enter at point $Q_2$ according to the modified reflection model. It can even leave the macro cell again from point $Q_3$ and re-enter at point $Q_4$, and so on. In this case, the probability that UE$_i$ will come into $C_k$ is given by
\begin{equation}\label{eq:pin_larg_f}
\begin{split}
  \Pr \{\textrm{come~in~indirectly} \}& \\
        =\sum_{m=1}^\infty\Pr\{ 2ml_1 & +\rho_1<X<2ml_1+\rho_2 \}
\end{split}
\end{equation}
where $l_1=|VQ_1|=\frac{1}{2}|Q_4Q_1|=\sqrt{L^2-d_0^2\sin^2\theta}$ is the half chord length.

If the user moves along direction $\overrightarrow{OQ_4}$ instead, this probability turns to be
\begin{equation}\label{eq:pin_larg_f1}
\begin{split}
  \Pr \{ \textrm{come~in~indirectly} \}&\\
        =\sum_{m=1}^\infty \Pr\{ 2ml_1 & -\rho_2<X<2ml_1-\rho_1 \}.
\end{split}
\end{equation}

Finally, by taking the average over $d_0$ and $\theta$, we obtain the incoming probability as follows.

\begin{proposition}\label{prop:pi}
   The average probability that UE$_i$ will come into the cell of interest is given by (\ref{rt:pin}), as shown on the top of next page.
\begin{figure*}
\hrulefill
   \begin{equation}\label{rt:pin}
   \begin{split}
     P_i(R,\Delta)=&\frac{2\Delta^\alpha}{\pi L^2}\int_{\max(R,\Delta-R)}^{\sqrt{\Delta^2+R^2}}
       \int_0^{\theta_2} \left(\frac{d_0}{\Delta^\alpha}-\frac{d_0}{\rho_2^\alpha} \right) d\theta d d_0
     +\frac{2\Delta^\alpha}{\pi L^2}\int_{\sqrt{\Delta^2+R^2}}^{\Delta+R}
       \int_{0}^{\theta_3} \left(\frac{d_0}{\rho_1^\alpha}-\frac{d_0}{\rho_2^\alpha} \right) d\theta d d_0\\
     +&\frac{2\Delta^\alpha}{\pi L^2}\int_{\Delta+R}^{L}
       \int_{0}^{\theta_3} \left(\frac{d_0}{\rho_1^\alpha}-\frac{d_0}{\rho_2^\alpha} \right) d\theta d d_0
     +\frac{2\Delta^\alpha}{\pi L^2}\int_{R}^{L}
       \int_{0}^{\theta_3} \left(\frac{d_0}{(2ml_1+\rho_1)^\alpha}-\frac{d_0}{(2ml_1+\rho_2)^\alpha} \right) d\theta d d_0\\
     +&\frac{2\Delta^\alpha}{\pi L^2}\int_{R}^{L}
       \int_{\pi-\theta_3}^{\pi} \left(\frac{d_0}{(2ml_1-\rho_1)^\alpha}-\frac{d_0}{(2ml_1-\rho_1)^\alpha} \right) d\theta d d_0.
   \end{split}
   \end{equation}
\hrulefill
\end{figure*}

\end{proposition}

\section{The Number of Users in $C_k$}\label{sec:4}
Let $C_k$ be an arbitrary small cell of interest and $R$ be its radius.
   Denote the number of users in $C_k$ at the beginning of time interval $[nT_s, (n+1)T_s]$ as $\xi_n$.
Due to user mobility, $\xi_n$ will be a random variable. In fact, $\xi_n$ dominates the number of interferers in the uplink. In this section, the stochastic characteristics of $\xi_n$ will be investigated.

\subsection{Queueing Model Formulation}
Assume that the users leave $C_k$ at the beginning of each time interval, which is denoted by $n^+$, and arrive at $C_k$ at the end of each interval, i.e., $n^-$. Specifically, $n^+=\lim_{t\rightarrow0^+}(nT+t)$ and $n^-=\lim_{t\rightarrow0^-}(nT+t)$.

 By its definition, $\xi_{n}$ is the number of users in $C_k$ at time $n^+$, where users arriving at $C_k$ between $(n^-,n]$ are included, and users leaving $C_k$ between $(n,n^+]$ are not included.

As shown in the previous section, $P_i(R,\Delta)$ is the average incoming probability and $P_o(R,\Delta)$ is the outgoing probabilities.
   Although the probability of coming into or moving out of a cell is different for different users and different locations, we assume that each user outside of $C_k$ may come into it with probability $P_i(R,\Delta)$, and each user in $C_k$ will leave with probability $P_o(R,\Delta)$, in the average sense.
In the following part, we will denoted $P_i(R,\Delta)$ and $P_o(R,\Delta)$ by $P_i$ and $P_o$ for notation simplicity.

If $\xi_{n-1}=k$, we will have $N-k$ users outside of $C_k$. Define the probability that there will be $j$ users coming into the cell at time $n^-$ as
\begin{equation}\label{df:nuj}\nonumber
\begin{split}
  \nu(j,N-k)
  &=C_{N-k}^j P_i^j (1-P_i^j)^{N-k-j}
\end{split}
\end{equation}
where  $C_n^k=\binom{n}{k}=\frac{n(n-1)(n-2)\cdots(n-k+1)}{k(k-1)(k-2)\cdots 1}$ is the combination function and $j=0,1,2,\cdots,N-k$.

Likewise, define the probability that there are $j$ users leaving the cell at time $n^+$ as
\begin{equation}\label{df:nuj}\nonumber
\begin{split}
  \mu(j,k)
  &=C_k^j P_o^j (1-P_o)^{k-j}
\end{split}
\end{equation}
where $j=0,1,2,\cdots,k$.

Therefore, the transition probability of the Markov chain $\{\xi_n,n\geq0\}$ is
\begin{equation}\label{dr:pij}\nonumber
\begin{split}
  p_{kj}=\Pr\{ \xi_n=j|\xi_{n-1}=k \}& \\
  =\sum_{r=\max(0,k-j)}^{\min(k,N-j)}& \mu(r,k)\nu(j+r-k,N-k)
\end{split}
\end{equation}
where $0\leq j \leq N-1$ and $0\leq k\leq N$.

Note that the upper limit of the summation can also be expressed by: $N-j=k-(j-(N-k))$.  That is, $j-(N-k)$ is  the gap to the goal of $j$ users on condition that all of other $N-k$ users will coming in, and can only be filled by users who will not leave $C_k$. Therefore, the maximum users can leave $C_k$ is at most $k-(j-(N-k))=N-j$.

With $\textbf{P}=\{p_{kj}\}_{(N+1)\times(N+1)}$, all the statistics of $\{\xi_n,n\geq0\}$ are hence determined.

\subsection{Stationary Distribution of $\xi_n$}
Since both $P_i$ and $P_o$ are positive and smaller than $1$, and the Markov chain has finite states, one can readily show that the Markov chain $\{\xi_n,n\geq0\}$ considered here has a stationary distribution $\bm{\pi}=\{\pi_j,0\leq j\leq N\}$.

Define the probability generating function (PGF) of $\bm \pi$ as
$\xi(z)=\sum_{j=0}^N\pi_j z^j$,
which will be given by the following theorem.
\begin{theorem}\label{th:1}
   The PGF of the stationary distribution of $\xi_n$, i.e., the number of users in $C_k$, is given by
   \begin{equation}\label{rt:xiPGF}
     \xi(z)=\left( \frac{P_i z}{P_i+P_o} + \frac{P_o}{P_i+P_o} \right)^N.
   \end{equation}
\end{theorem}

\begin{proof}
Firstly, the stationary distribution $\bm\pi$ satisfies the following equations.
\begin{equation}\label{eq:stat}\nonumber
  \bm{\pi}\textbf{P}=\bm{\pi},\quad \bm{\pi}\textbf{e}=1
\end{equation}
where $\textbf{e}$ is a row vector of ones.

For the $j$-th element of stationary distribution $\pi_j$, we have
\begin{equation}\label{eq:statieq}\nonumber
  \pi_j=\sum_{k=0}^N \pi_k p_{kj},\quad j=0,1,\cdots,N.
\end{equation}

By multiplying $z^j$ on both sides and take the summation from $1$ to $N$, we have
\begin{equation}\label{dr:xizzzz}\nonumber
  \begin{split}
    & \xi(z) = \sum_{j=1}^N \pi_j z^j \\
      \stackrel{(a)}{=} &\sum_{k=0}^N \pi_k \left(\sum_{r=0}^{k} \sum_{j=k-r}^{N-r} z^j \mu(r,k)\nu(j+r-k,N-k)\right)\\
      \stackrel{(b)}{=} &\sum_{k=0}^N \pi_k \left(\sum_{r=0}^{k} \sum_{i=0}^{N-k} z^{i+k-r} \mu(r,k)\nu(i,N-k)\right)\\
      = &\sum_{k=0}^N \pi_k (P_o+(1-P_o)z)^k (P_iz+1-P_i)^{N-k}\\
      = &(P_iz+1-P_i)^{N} \xi\left(\frac{P_o+(1-P_o)z}{P_iz+1-P_i}\right),
  \end{split}
\end{equation}
where the order of the summations is changed in (a) and variable substitution $i=j+r-k$ is used in (b).

Then the theorem is established by solving $\xi(z)$ from the above equation.
\end{proof}

\begin{remark}
   Using the polynomial expansion to $\xi(z)$ we will have
   \begin{equation}\label{xizexpan}\nonumber
     \xi(z)=\sum_{j=0}^N C_N^j \left(\frac{P_i z}{P_i+P_o}\right)^j \left(\frac{P_o}{P_i+P_o}\right)^{N-j}.
   \end{equation}

   It is clear that $\pi_j=C_N^i \left(\frac{P_i}{P_i+P_o}\right)^j \left( \frac{P_o}{P_i+P_o}\right)^{N-j}$, which means that the number of users in $C_k$, i.e., $\xi_n$ is a Binomial distributed random number in the limit sense.
\end{remark}

\begin{remark}
   Denote $\eta=\frac{P_i}{P_i+P_o}$ and $\lambda=N\eta$. It is well known that Binomial distribution can be approximated by Poisson distribution when $N$ is very large and $\eta$ is very small, which will make further analysis easier.
\end{remark}

It is known that the mean and variance of a random variable $X$ are related with its PGF $G_X(z)$ through following equations
\begin{equation}\label{eq:EDvsPGF}\nonumber
  \begin{split}
    \mathbb{E}[X] &= G_X'(z)|_{z=1} \\
    \mathbb{D}[X] &= G_X''(z)-\left(G_X'(z)\right)^2+G_X'(z)|_{z=1}.
  \end{split}
\end{equation}

Then the statistics of the number of users in $C_k$ are given by the following proposition.

\begin{proposition}\label{prop:edxi}
   The average and the variance of number of users in $C_k$ are given by, respectively
   \begin{equation}\label{rt:EDxi}
     \begin{split}
       \mathbb{E}[\xi_n] & = \frac{NP_i}{P_i+P_o} \\
       \mathbb{D}[\xi_n] & = \frac{NP_iP_o}{(P_i+P_o)^2}.
     \end{split}
   \end{equation}
\end{proposition}

\begin{remark}
   Note that both $P_i$ and $P_o$ are functions of basic step length $\Delta$, which is an index of moving velocity. Therefore, both $\mathbb{E}[\xi_n]$ and $\mathbb{D}[\xi_n]$ are also functions of user velocity.
\end{remark}

\section{The Randomness of Uplink Interference}\label{sec:5}
Usually, M-UEs transmit power of is  much higher than that of H-UEs since M-UEs are very far from the M-eNB.
 As a result, M-UEs' uplink signals will be a great interference to the H-UEs nearby.
In addition, this uplink interference will change randomly along time due to the following reasons.

First, each interferer is located randomly and moves randomly.
 Therefore, their distances to the interfered H-eNB are also random, which introduces uncertainty to the interference.
Second, interfering signals suffer from small scale fading, which vary quickly along time.
   Last but not the least, the number of interferers is random due to user mobility.
Particularly, its fluctuation is further accelerated by the miniaturization of cells. As a result, the uplink interference also has a `fading' property.

However, it should be noted that the `fading' of the uplink interference caused by users' mobility is a kind of large scale fading and a slow fading.
   Generally speaking, the velocity of a user is 3 km/h for pedestrians and about 120 km/h if the user is in a vehicle.
Therefore, the flight time is relatively large, which makes the fluctuation of the uplink interference much slower than small scale fading.

In the following part, the fading property of uplink interference will be characterized in terms of distribution and statistic moments, based on which the impact of user mobility on uplink interference can be revealed.

\subsection{Uplink Interference to CSG Femto cells}
In a CSG femto cell $C_k$ with radius $R$, only some authenticated users within its cell coverage are allowed to communicate with the H-eNB. Those unauthenticated  UEs have to be linked to the M-eNB, even if it is in $C_k$. As shown in Fig. \ref{fig:net model}, each UE within the circle of interfering radius $R_I$, which is referred to as $C'$, is an interferer to femto UEs in $C_k$.

Let $\xi_{\textrm{iii}}$ be the number of M-UEs in the interfering circle $C'$ in the $n$-th time interval. Thus it is a Binomial distributed random variable with its PGF $\xi_3(z)$ given by (\ref{rt:xiPGF}), Theorem \ref{th:1}.

Denote the distance between M-UE$_j$ and the femto e-NB as $d_j$ where $d_j>1$ m is assumed.
   Let $\gamma_j$ be the small scale fading power gain.
Let $F_\gamma(x)$  be its cumulative distribution function (CDF), $P_\gamma=\mathbb{E}_\gamma[\gamma_j]$ is the average power gain and $P_\gamma^{(2)}=\mathbb{E}_\gamma[\gamma_j^2]$ is the second order moment.
   Denote M-UEs' transmit power as $P_t^m$, it is seen that the instantaneous interference can be expressed as $I_{cj}= \frac{\gamma_j P_t^m}{d_j^\beta}$, with its moments given by following proposition.

\begin{proposition}\label{prop:muc}
The first and second order moments of the interference from a uniformly distributed M-UE within the interfering circle are
\begin{equation}\label{rt:edIcj}
  \begin{split}
    \mu_{c} &= \mathbb{E}[I_{cj}]=\frac{2P_t^mP_\gamma (R_I^{\beta-2}-1)} {(\beta-2)R_I^{\beta-2}(R_I^2-1)} \\
    \mu_{c}^{(2)} &= \mathbb{E}[I_{cj}^2]=\frac{(P_t^m)^2P_\gamma^{(2)} (R_I^{2\beta-2}-1)} {(\beta-1)R_I^{2\beta-2}(R_I^2-1)}.
  \end{split}
\end{equation}
\end{proposition}

\begin{proof}
    Since M-UE$_j$ is uniformly distributed within the interference circle, the probability that its distance to the H-eNB is less than $x$ is $\Pr\{d_j<x\}=\frac{x^2-1}{R_I^2-1}$. It is readily obtained that the {\em pdf} of $d_j$ is
   $f_d(x)=\frac{2x}{R_I^2-1},~ x\in[1,R_I]$.
        Next, the first and second order moment of the interference can be obtained readily by taking its average over $d_j$ and $\gamma_j$.
%
%
\end{proof}

\begin{remark} In the case of $\beta=2$, (\ref{rt:edIcj}) holds in the limitation sense. That is,
\begin{equation}\label{dr:limibeta2}\nonumber
  \begin{split}
    \mu_{c} &= \lim_{\beta\rightarrow2} \frac{2P_t^mP_\gamma (R_I^{\beta-2}-1)} {(\beta-2)R_I^{\beta-2}(R_I^2-1)}\\
      &= \frac{2P_t^mP_\gamma } {R_I^2-1}  \lim_{\beta\rightarrow2} \frac{R_I^{\beta-2}-1}{(\beta-2)}\frac{1}{R_I^{\beta-2}}
      = \frac{2P_t^mP_\gamma } {R_I^2-1} \ln R_I.
  \end{split}
\end{equation}
\end{remark}

\begin{remark} Actually, both $\mu_{c}$ and $\mu_{c}^{(2)}$ are decreasing with pathloss exponent $\beta$, which can be proved by checking their derivatives versus $\beta$.

\end{remark}

Since there are $\xi_{\textrm{iii}}$ M-UEs within the interfering circle, the total interference will be
\begin{equation}\label{eq:Ic}\nonumber
  I_{c}=\sum_{j=1}^{\xi_{\textrm{iii}}}I_{cj}=\sum_{j=1}^{\xi_{\textrm{iii}}} \frac{\gamma_j P_t^m}{d_j^\beta}.
\end{equation}

Define $G_{I_{cj}}(s)=\mathbb{E}[e^{sI_{cj}}]$
as the Moment generating function (MGF) of each individual interference. Then
the MGF of the total interference and its average and variance are summarized by the following theorem.

\begin{theorem}\label{th:closemgfED}
The MGF of the uplink interference to CSG femto cell UEs is
\begin{equation}\label{rt:upImgfc}
  G_{I_c}(s)=\left( \frac{P_i G_{I_{cj}}(s)}{P_i+P_o} + \frac{P_o}{P_i+P_o} \right)^N.
\end{equation}

Its average and variance are given by, respectively
\begin{equation}\label{rt:interclose}
   \begin{split}
       \mathbb{E}[I_c] & = \frac{NP_i}{P_i+P_o}\mu_c \\
       \mathbb{D}[I_c] & = \frac{NP_i}{P_i+P_o}\mu_c^{(2)}-\frac{NP_i^2}{(P_i+P_o)^2}\mu_c^2.
   \end{split}
\end{equation}
\end{theorem}

\begin{proof}
    By its definition, one has
    \begin{equation}\label{dr:mgfxi3}\nonumber
    \begin{split}
      G_{I_c}(s)&=\mathbb{E}[e^{sI_c}]=\mathbb{E}[e^{s\sum_{j=1}^{\xi_{\textrm{iii}}} I_{cj}}]\\
                &=\sum_{k=0}^N \Pr\{\xi_{\textrm{iii}}=k\}\left(\mathbb{E}[e^{sI_{cj}}]\right)^k
                =\xi\left(G_{I_{cj}}(s)\right)
    \end{split}
    \end{equation}
    where $\xi(z)$ was given by (\ref{rt:xiPGF}). This proves (\ref{rt:upImgfc}).

    Then the average uplink interference will be
    \begin{equation}\label{dr:averageIc}\nonumber
            \mathbb{E}[I_c]=G_{I_c}'(s)|_{s=0}
            =\frac{NP_i}{P_i+P_o}\mu_c.
    \end{equation}

    Similarly, its second moment is
    \begin{equation}\label{dr:averageIc2}\nonumber
      \begin{split}
        &\mathbb{E}[I_c^2]=G_{I_c}''(s)|_{s=0}
            =\frac{N(N-1)P_i^2}{(P_i+P_o)^2}\mu_c^2+\frac{NP_i}{P_i+P_o}\mu_c^{(2)}
      \end{split}
    \end{equation}
    where $\mu_c$ and $\mu_c^{(2)}$ are given by Proposition \ref{prop:muc}.

    Therefore, the variance of uplink interference is
    \begin{equation}\label{dr:averageIc2}\nonumber
      \begin{split}
      \mathbb{D}[I_c]&=\mathbb{E}[I_c^2]-\mathbb{E}^2[I_c]
                =\frac{NP_i}{P_i+P_o}\mu_c^{(2)}-\frac{NP_i^2}{(P_i+P_o)^2}\mu_c^2.
                \end{split}
    \end{equation}

\end{proof}

\begin{remark}
    It is known that  the  pdf of a random variable is completely determined by its MGF \cite{Fan-SCT}. Thus Theorem \ref{th:closemgfED} gives a full characterization of the uplink interference to a CSG femto cell.
       Explicit expressions for $G_{I_{cj}}(s)$ can also be obtained for any given $f_\gamma(x)$.
\end{remark}

\begin{remark}
    Two key parameters for the results in Theorem \ref{th:closemgfED} are $R_I$ and $\Delta$. While $R_I$ specifies the interfering area, $\Delta$ indicates the mobility of users. Therefore, this theorem has presented how user mobility affects the uplink interference.
\end{remark}

\subsection{Uplink Interference to OSG Femto cells}
OSG femto or pico/relay cells will admit every users coming into their coverage. Therefore, only M-UEs outside the cell but within the interfering radius will cause interference.

Assume there are $\xi_{\textrm{iii}}$ users in all within the circular area $C'$ of radius $R_I$, in which $\xi_{\textrm{i}}$ users locates within $C_k$. Thus the number of interferers in the interfering ring is $\xi_{\textrm{ii}}=\xi_{\textrm{iii}}-\xi_{\textrm{i}}$.

Let $d_j\in (R,R_I)$ be the distance between an interferer and the H-eNB.
   Its interference to H-UEs is $I_{oj}=\frac{\gamma_jP_t^m}{d_j^\beta}$ and the total interference is
\begin{equation}\label{eq:Io}\nonumber
  I_{o}=\sum_{j=1}^{\xi_{\textrm{ii}}}I_{oj}=\sum_{j=1}^{\xi_{\textrm{ii}}} \frac{\gamma_j P_t^m}{d_j^\beta}.
\end{equation}

First, the first and second moments of $I_{oj}$ are given by the following proposition.

\begin{proposition}\label{prop:nuo}
The first and second order moments of the interference from a uniformly located M-UE within the interfering ring are
\begin{equation}\label{rt:edIoj}
  \begin{split}
    \nu_{o} &= \mathbb{E}[I_{oj}]=\frac{2P_t^mP_\gamma (R_I^{\beta-2}-R^{\beta-2})} {(\beta-2)R_I^{\beta-2}R^{\beta-2}(R_I^2-R^2)} \\
    \nu_{o}^{(2)} &= \mathbb{E}[I_{oj}^2]=\frac{(P_t^m)^2P_\gamma^{(2)} (R_I^{2\beta-2}-R^{2\beta-2})} {(\beta-1)R_I^{2\beta-2}R^{2\beta-2}(R_I^2-R^2)}.
  \end{split}
\end{equation}
\end{proposition}

The proof of Proposition \ref{prop:nuo} is similar to that of Proposition \ref{prop:muc} and is omitted here. It can be proved that $\nu_{o}$ and $\nu_{o}^{(2)}$ are also decreasing with $\beta$.

Define $\varphi(x)=\frac{R_I^{\beta-2}-x^{\beta-2}}{x^{\beta-2}}=\left(\frac{R_I}{x}\right)^{\beta-2}-1$ for $x\geq1$, it is clear that $\varphi(x)$ is decreasing with $x$ and $\varphi(1)>\varphi(R)$. By comparing (\ref{rt:edIcj}) and (\ref{rt:edIoj}), we have $\mu_c>\nu_o$ and $\mu_c^{(2)}>\nu_o^{(2)}$.

Define $G_{I_{oj}}(s)=\mathbb{E}[e^{sI_{cj}}]$
as the MGF of the instant interference from M-UE$_j$,
the MGF of $I_o$ and its average and variance are given by the following theorem.

\begin{theorem}\label{th:openmgfED}
The MGF of the uplink interference to OSG femto cell UEs is
\begin{equation}\label{rt:upImgfo}
  G_{I_o}(s)=\left( \frac{P_i (qG_{I_{oj}}(s)+1-q)}{P_i+P_o} + \frac{P_o}{P_i+P_o} \right)^N.
\end{equation}

The average and variance are given by, respectively
\begin{equation}\label{rt:interopen}
   \begin{split}
       \mathbb{E}[I_o] & = \frac{NP_iq}{P_i+P_o}\nu_o \\
       \mathbb{D}[I_o] & = \frac{NP_iq}{P_i+P_o}\nu_c^{(2)}-\frac{NP_i^2q^2}{(P_i+P_o)^2}\nu_c^2
   \end{split}
\end{equation}
where $q=1-\frac{R^2}{R_I^2}$, $P_i$ and $P_o$ are calculated with $R_I$ and $\Delta$.
\end{theorem}

\begin{proof}
    For any user who has moved into the interfering circle $C'$, its location is uniformly distributed in the area by Lemma \ref{lem:unif_distri}. Thus its probability of lying in the interfere ring is $q=\frac{(\pi R_I^2-\pi R^2)}{\pi R_I^2}=1-\frac{R^2}{R_I^2}$. Then the probability that there are $k$ users in the interfering ring is
    \begin{equation}\label{dr:pxi2k}\nonumber
    \begin{split}
      \Pr\{\xi_{\textrm{ii}}=k\}  
         &=\sum_{i=0}^{N-k}\Pr\{\xi_{\textrm{i}}=i,\xi_{\textrm{iii}}=k+i\} \\
         &=\sum_{i=0}^{N-k}\Pr\{\xi_{\textrm{iii}}=k+i\}C_{k+i}^kq^k(1-q)^i.
    \end{split}
    \end{equation}

    Next, the MGF of $\xi_{\textrm{ii}}$ is
    \begin{equation}\label{dr:xi2mgf}\nonumber
    \begin{split}
      G_{\xi_{\textrm{ii}}}(z)&=\sum_{k=0}^N z^k \Pr\{\xi_{\textrm{ii}}=k\}\\
        &=\sum_{k=0}^N z^k \sum_{i=0}^{N-k}\Pr\{\xi_{\textrm{iii}}=k+i\}C_{k+i}^kq^k(1-q)^i\\
        &\stackrel{(a)}{=}\sum_{k=0}^N z^k \sum_{j=k}^{N}\Pr\{\xi_{\textrm{iii}}=j\}C_{j}^k q^k(1-q)^{j-k}\\
        &=\xi(qz+1-q)
    \end{split}
    \end{equation}
    where variable substitution $j=k+i$ is used in (a).

    Then the MGF of total interference $I_o$ will be
    \begin{equation}\label{dr:mgfIooo}\nonumber
      \begin{split}
            G_{I_o}(s)&=\mathbb{E}[e^{sI_o}]=\mathbb{E}[e^{s\sum_{j=1}^{\xi_{\textrm{ii}}}I_{oj}}]\\
                    &=\sum_{k=1}^N \Pr\{\xi_{\textrm{ii}}=k\}\left(\mathbb{E}[e^{sI_{oj}}]\right)^k
                    =G_{\xi_{\textrm{ii}}}(G_{I_{oj}}(s))
      \end{split}
    \end{equation}
    which proves (\ref{rt:upImgfo}).

    The average interference is
    \begin{equation}\label{dr:averageIo}\nonumber
    \begin{split}
      \mathbb{E}[I_o]&=G_{I_o}'(s)|_{s=0}=G_{\xi_{\textrm{ii}}}'(G_{I_{oj}}(s))|_{s=0}
                =\frac{N P_i q}{P_i+P_o}\nu_o.
    \end{split}
    \end{equation}

    The second moment of $I_o$ can be obtained by
    \begin{equation}\label{dr:averageIo2}\nonumber
    \begin{split}
      \mathbb{E}[I_o^2]=&G_{I_o}''(s)|_{s=0} \\ 
                =&\frac{N(N-1)P_i^2 q^2}{(P_i+P_o)^2}\nu_o^2 + \frac{N P_i q}{P_i+P_o}\nu_o^{(2)}.
    \end{split}
    \end{equation}

    Finally, (\ref{rt:interopen}) will be proved by using $\mathbb{D}[I_o]=\mathbb{E}[I_o^2]-\mathbb{E}^2[I_o]$.
\end{proof}

\begin{remark}
   Recall that $\nu_o<\mu_c$ is true. Thus we have $\mathbb{E}[I_o]<\mathbb{E}[I_c]$, which  means that an OSG femto cell will suffer less cross-tier uplink interference than a CSG femto cell on average.
\end{remark}

\section{Success Probability and Average Rate}\label{sec:6}

As an application of the developed interference model, this section will characterize the system performance as a function of the random signal-to-interference-plus-noise ratio
(SINR), which is given by
\begin{equation}\label{df:sinr}
  \rho(\kappa)=\frac{P_r^f(\kappa)}{I+P_n}.
\end{equation}

 In the above equation, $P_n$ is the noise power, $I$ is the uplink interference and $I=I_c$ for CSG femto cells and $I=I_o$ for OSG femto cells. $P_r^f(\kappa)$ is the received power at the H-eNB from a H-UE which is at a distance of $d=\kappa R$ away, $\kappa\in(\frac1R,1)$.
    Let $P_t^h$ be H-UEs' transmit power, we have $P_r^f(\kappa)=\frac{\gamma P_t^h}{\kappa^\beta R^\beta}$.

It is assumed that $\gamma$ follows the negative exponential distribution, namely the Rayleigh fading model with
$f_\gamma(x)=\frac{1}{P_\gamma} e^{-\frac{x}{P_\gamma}}$.

SINR is a useful quantification for performance analysis in cellular systems since system performance is
usually interference limited, especially for users at the cell edge.
   In our formulation, it is assumed that no pre-coding or multi-user detection are used at the H-eNB. Thus signals from unexpected users contributes to interference only.

Two metrics of performance are evaluated: the success probability defined as $\Pr\{\rho(\kappa)\geq T\}$ where $T$ is
a given threshold,  and the average achievable rate, which is given by
\begin{equation}\label{df:averageC}
  C(\kappa)=W\mathbb{E}\ln(1+\rho(\kappa)),
\end{equation}
where $W$ is the system bandwidth.

While success probability measures the impact of interference and channel fading on transmission reliability, the achievable rate indicates the cell's transmission efficiency.
   Besides, the outage behavior $P_{outage}=\Pr\{\rho(\kappa)\leq T\}$ can also be obtained along the same line.

\subsection{Success Probability}
The success probability of both CSG and OSG femto cell UEs can be summarized in the following proposition.

\begin{proposition}\label{prop:successProba}
The success probability of a femto user  is given by
\begin{equation}\label{prop:sucPro}\nonumber
\Pr\{\rho(\kappa)\geq T\}=\exp\left(\frac{-\kappa^\beta R^\beta P_n T}{P_t^h P_\gamma}\right) G_I\left(\frac{-\kappa^\beta R^\beta T }{P_t^h P_\gamma} \right)
\end{equation}
   where $G_I(s)=G_{I_c}(s)$ is given by (\ref{rt:upImgfc}), for CSG femto cells and $G_I(s)=G_{I_o}(s)$ is given by (\ref{rt:upImgfo}), for OSG femto cells.
\end{proposition}

\begin{proof}
    Whether a user can access to the H-eNB successfully or not depends both on the instant channel gain and the instant uplink interference. Thus the success probability will be
    \begin{equation}\label{dr:sucprobbs}\nonumber
    \begin{split}
      \Pr \left\{ \rho(\kappa)\geq T  \right\} 
          =&\int_0^\infty f_I(x)dx \int_{\frac{\kappa^\beta R^\beta}{P_t^h}(I+P_n) T}^\infty f_\gamma(y)dy \\
          =&\exp\left(\frac{-\kappa^\beta R^\beta P_n T}{P_t^hP_\gamma}\right) G_I \left(\frac{-\kappa^\beta R^\beta  T}{P_t^hP_\gamma}\right)
    \end{split}
    \end{equation}
    where $G_I(s)=G_{I_c}(s)$ or $G_{I_o}(s)$ is the MGF of the interference from M-UEs to CSG or OSG femto cells, given by (\ref{rt:upImgfc}) and (\ref{rt:upImgfo}), respectively.
\end{proof}

In both Theorem \ref{th:closemgfED} and Theorem \ref{th:openmgfED}, the MGFs of uplink interference are presented in terms of the MGF of the interference from a certain macro interferer $j$, i.e., $G_{I_{cj}}(s)$ and $G_{I_{oj}}(s)$, respectively.
   For Rayleigh fading and commonly used pathloss exponents, one has
\begin{equation}\label{dr:gicjmgf}\nonumber
\begin{split}
  G_{I_{cj}}(s)&=\mathbb{E}[e^{sI_{cj}}]
     =1+ \frac{2P_t^mP_\gamma s}{R_I^2-1} \int_1^{R_I} \frac{y}{y^\beta-P_t^mP_\gamma s} dy\\
  G_{I_{oj}}(s)&=\mathbb{E}[e^{sI_{oj}}]
     =1+ \frac{2P_t^mP_\gamma s}{R_I^2-R^2} \int_R^{R_I} \frac{y}{y^\beta-P_t^mP_\gamma s} dy.
\end{split}
\end{equation}

Closed form expressions can be obtained for some special cases such as $\beta=2$, $\beta=4$.
\begin{equation}\nonumber
  \begin{split}
   G_{I_{oj}}(s) &= 1+ \frac{P_t^mP_\gamma s}{R_I^2-R^2} \ln\frac{R_I^2-P_t^mP_\gamma s}{R^2-P_t^mP_\gamma s}, \quad \beta=2;\\
   G_{I_{oj}}(s)  &= 1+ \frac{\sqrt{P_t^mP_\gamma s}}{2(R_I^2-R^2)} \\
            & \quad \cdot \ln \frac{(R_I^2-\sqrt{P_t^mP_\gamma s})(R^2+\sqrt{P_t^mP_\gamma s})} {(R_I^2+\sqrt{P_t^mP_\gamma s})(R^2-\sqrt{P_t^mP_\gamma s})}, \quad\beta=4,\\
   G_{I_{cj}}(s) &= 1+ \frac{P_t^mP_\gamma s}{R_I^2-1} \ln\frac{R_I^2-P_t^mP_\gamma s}{1-P_t^mP_\gamma s},\quad \beta=2;\\
    G_{I_{cj}}(s) & = 1+ \frac{\sqrt{P_t^mP_\gamma s}}{2(R_I^2-1)} \\
            & \quad \cdot \ln \frac{(R_I^2-\sqrt{P_t^mP_\gamma s})(1+\sqrt{P_t^mP_\gamma s})} {(R_I^2+\sqrt{P_t^mP_\gamma s})(1-\sqrt{P_t^mP_\gamma s})}, \quad \beta=4.
  \end{split}
\end{equation}

There are three key parameters involved in Proposition \ref{prop:successProba}, namely $\Delta$, $R_I$ and $\kappa$.

First, $\Delta$ is an indicator of users' mobility and is the biggest difference between our interference model from others.
   Second, $R_I$ is the interfering radius which determines how many M-UEs will cause interference.
Specifically, the influence of $\Delta$ and $R_I$ are implied through $P_o$ and $P_i$ in Proposition \ref{prop:po}, \ref{prop:pi} and in Theorem \ref{th:closemgfED}, \ref{th:openmgfED}.
   Finally, $\kappa=\frac{d}{R}$ specifies the distance between the user and the H-eNB, and will dominate the distant-dependent success probability function $\Pr \left\{ \rho(\kappa)\geq T  \right\}$.

\subsection{Average rate}
Average rate is another important evaluation of UEs' performance, especially for those edge UEs.  The average rate of a user at position $\kappa$ is obtained by averaging over the random interference and the fading channel gain, and is
summarized in the following proposition.

\begin{proposition}
    The average rate of a femto user at position $\kappa=\frac{d}{R}$ is
    \begin{equation}\label{rt:prop7}
        \begin{split}
            \mathbb{E}[C(\kappa)]=W\int_0^\infty \frac{P_t^h P_\gamma e^{-P_nx}}{\kappa^\beta R^\beta+P_t^h P_\gamma x}  G_I(-x)dx
        \end{split}
    \end{equation}
    where $G_I(s)=G_{I_c}(s)$  for CSG femto cells and $G_I(s)=G_{I_o}(s)$  for OSG femto cells.
\end{proposition}

\begin{proof}
    The randomness of $C(\kappa)$ comes from the fading of the channel as well as the randomness of uplink interference. Its CDF is given by
    \begin{equation}\label{dr:cdfckappa}\nonumber
    \begin{split}
      F_C(x)&=\Pr\{C(\kappa)\leq x\}   \\
         &=\int_0^\infty f_I(z)dz \int_0^{\frac{\kappa^\beta R^\beta}{P_t^h}(e^{\frac xW}-1)(z+P_n)} f_\gamma(y) dy\\
         &=1-\exp\left(\frac{-\kappa^\beta R^\beta P_n }{P_t^h P_\gamma}(e^\frac{x}{W}-1)\right)   \\
         & \qquad~~\cdot  G_I\left( \frac{-\kappa^\beta R^\beta }{P_t^h P_\gamma}(e^\frac{x}{W}-1)\right).
    \end{split}
    \end{equation}

    Then the average rate will be
    \begin{equation}\label{dr:averagerate}\nonumber
      \begin{split}
        \mathbb{E}[C(\kappa)]&=\int_0^\infty x dF_C(x)\\
           &=W\int_0^\infty \frac{P_t^h P_\gamma e^{-P_n x}}{\kappa^\beta R^\beta+P_t^h P_\gamma x} G_I(-x)dx.
      \end{split}
    \end{equation}
\end{proof}

Note that this can be readily calculated numerically.

\section{Simulation Results}\label{sec:7}
In a heterogeneous network as shown in Fig. \ref{fig:net model}, we consider one of the macro cells in the network.
   The radius of the macro cell is $L=500$ m.
There are $N=10000$ users in the macro cell. Therefore, we have one user for every $78.5$ m$^2$.
   Femto/pico/relay cells are also placed in the network, which are covered by low power H-eNBs.
   Although the both M-UEs and H-UEs have the same maximum transmit power, it is reasonable to assume that after power control, the actual transmit power of M-UEs is $P_t^m=20$ dBm and the transmit power of those H-UEs is $P_t^h=-3$ dBm.
Assume that the radius of a femto cell of interest is $R=60$ m, and the interfering radius is $R_I=120$ m.
Let $W=5$ MHz be the system bandwidth and $N_0=3.98107\times10^{-18}$ W/Hz be the noise power spectrum density.
Thus the noise power is $P_n=WN_0$.
  Assume that received signal at each eNB suffers from Rayleigh fading. In this case, the channel power gain follows the exponential distribution $f_\gamma(x)=\frac{1}{P_\gamma} e^{-\frac{x}{P_\gamma}}$, where $P_\gamma=1$ is the average power gain.
If no otherwise specified, the L\'{e}vy flight parameter is $\alpha=0.6$ and the pathloss exponent is $\beta=2$.
Suppose the flight time is $T_s=1$ s, then the user velocity will be 3 to 120 km/h when we set $0.833\leq\Delta\leq 33.3$, and if the flight length equals to one basic step length.
 Note that, users' instantaneous velocity depends both on $\Delta$ and instantaneous flight length. Although $\Delta$ is a constant in the simulation, users' instantaneous velocity will be random, which can mimic the random behaviors of users.

In the Monte Carlo simulation, the end points of each flight are determined in the following way. As shown in Fig. \ref{fig:reflec}, assume the origin of the polar coordinate is at the center of the macro cell (i.e., point $M$).
Assume UE$_i$ starts from point $O$ (i.e., $(\rho, \theta)$), towards point $Q_1$ and stops at point $O',~(\rho',\theta')$.
   Denote the  length and the direction of UE$_i$' flight as $(f, \gamma)$.
Denote $l_1=|VQ_1|=\sqrt{L^2-\rho^2\sin^2(\gamma-\theta)}$, $l_3=|OQ_1|=l_1-\rho \cos(\gamma-\theta)$ and $l_4=f-l_3$.
   Then $l_4$ is the remaining flight length after UE$_i$ has reached point $Q_1$. Let $m=\lfloor \frac{l_4}{2l_1} \rfloor$ be the number of times that UE$_i$ moves out of $C_k$, where $m=-1$ means that it will not cross the macro cell edge.
Denote $l_5 =|\mod (l_4, 2l_1)|$ as the distance between the end point of UE$_i$ and the point on the edge from which it re-enters the macro cell for the last time.
   Define $\gamma_1=\arctan \frac{l_1-l_4}{\sqrt{L^2-l_1^2}}$.
Then we have $\rho'=\sqrt{L^2-l_1^2+(l_1-l_4)^2}$ and $\theta'=\mod\left( \gamma - (-1)^{m+1}\mathrm{sign}(\gamma-\theta) \frac{\pi}{2} - (-1)^{m+1}\mathrm{sign}(\gamma-\theta) \mathrm{sign} \right.$ $\left. (l_1-l_5) \gamma_1 , 2\pi \right)$. It can be verified that this calculation of $\rho$ and $\theta$ is applicable to any $(\rho,\theta)$ and $(f,\gamma)$.

\begin{figure}[h]
\centering
\includegraphics[width=3.0in]{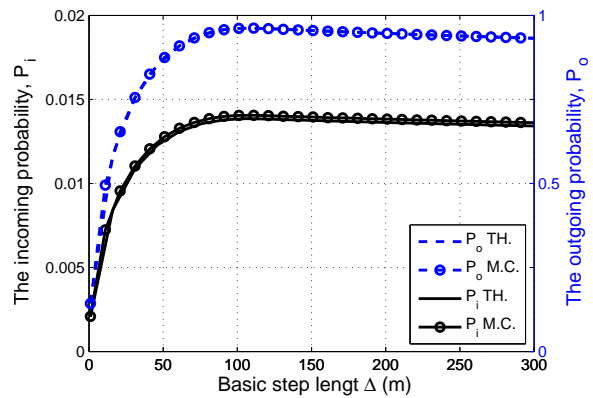}
\caption{The outgoing/incoming probability versus basic step length $\Delta$. The outgoing probability $P_o$ is shown by the dashed line and corresponds to the vertical axis on the right. The incoming probability $P_i$ corresponds to the solid line and the vertical axis on the left. Theoretical results and Monte Carlo Results are labeled by `TH.' and `M.C.' respectively. } \label{fig:pipo}
\end{figure}

First, the average incoming probability and outgoing probability are presented in Fig. \ref{fig:pipo}.
   To show a full picture of the relationships among them, let $\Delta$ range from $1$ to $300$ m.
The simulation was implemented independently for each $\Delta$. In each run of simulation, $\Delta$ is fixed and each user moves $10^6$ steps.
It is seen that as basic step length $\Delta$ increases, both $P_i$ and $P_o$  will increase  first and then decrease slightly.
   When $\Delta$ is small, users's mobility is so weak that they have little chance to move into or move out of the small cell $C_k$.
As $\Delta$ is increased, more users around $C_k$ will have the chance to enter. Thus $P_i$ will also be increased when $\Delta$ is increased.  However, the probability for a user to enter $C_k$ will decrease if $\Delta$ is increased when it has been  very large. This is because the user has more chance to step over $C_k$ rather than come into it. In summary, $P_i$ will increase first and then decrease when $\Delta$ is increased gradually.   It is seen that $P_o$ also decreases a little when $\Delta$ becomes very large, in which case a user can move out of $C_k$ as well as the macro cell, and may return to $C_k$ according to the modified reflection model.
It is seen that the theoretical results on $P_o$ coincide with the Monte Carlo results very well.
But there is a small gap between the results for $P_i$. This is because the calculation of $P_i$ according to (\ref{rt:pin}) involves a lot of very small valued integrations, which makes the result slightly smaller due to limited calculation accuracy in the simulation.

\begin{figure}[h]
\centering
\includegraphics[width=3.0in]{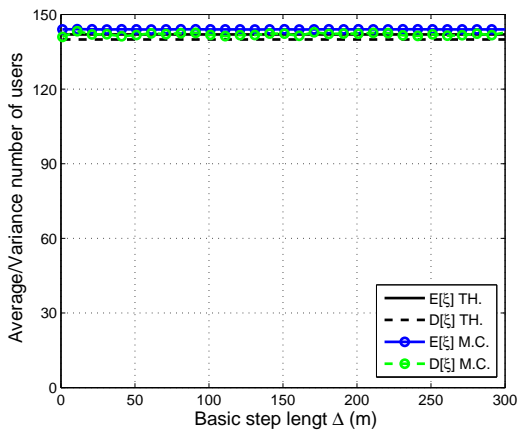}
\caption{The average and variance of the number of users in $C_k$ versus. $\Delta$.}\label{fig:edxi}
\end{figure}

The statistics of the number of users in a small cell of interest $C_k$ are presented in Fig. \ref{fig:edxi}.
It is seen that both the average $\mathbb{E}[\xi_n]$ and the variance $\mathbb{D}[\xi_n]$ of the number of users in $C_k$ is constant as $\Delta$ changes. This means that they are independent of the user velocity. That is, although the outgoing probability and incoming probability are closely related to user velocity (as shown in Fig. \ref{fig:pipo}), $\mathbb{E}[\xi_n]$ and the variance $\mathbb{D}[\xi_n]$ will not change.

Modeling the small scale fading by Rayleigh distribution, the statistics of the total interference of CSG and OSG femto cells are evaluated according to Theorem \ref{th:closemgfED} and Theorem \ref{th:openmgfED}, as shown in Fig. \ref{fig:edico}.
It is seen that the average interference will not change with basic step length $\Delta$ either, neither is its variance, for both CSG and OSG femto cells.
  Generally speaking, the interference from users with low mobility and that from users with high mobility are the same in the strength.
   The interference from users with high mobility, however, will change more quickly in the time domain.

\begin{figure}[h]
\centering
\includegraphics[width=3.0in]{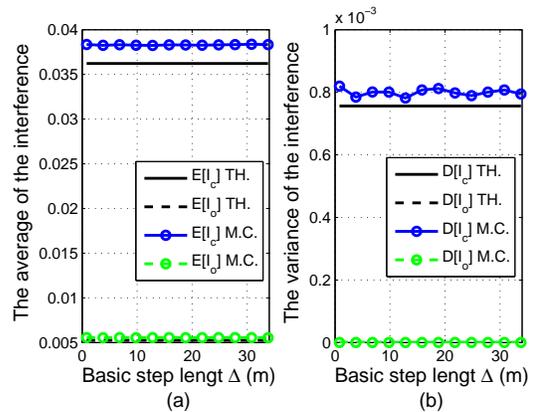}
\caption{The average and variance of the interference of CSG and OSG femto cells (W) v.s. $\Delta$, presented by (a) and (b), respectively.}\label{fig:edico}
\end{figure}

In addition, it is shown that the uplink interference to an OSG femto cell is much smaller than that to a CSG femto cell.
   In a CSG cell, interferers (unauthenticated users) may locate inside the cell.
In a OSG cells, however, only users outside of the cell but within the interfering circle will cause interference, which means that the number of interferers are fewer and the interfering distance is larger.
    In this case, the number of interferers is smaller and the distance between them and the H-eNB is larger.  Therefore, OSG femto cells are more beneficial to the network performance.
This makes sense if the radius of small cells is comparable to the interfering radius, in which case it is most likely to have only one small cell (denoted as $C_k$) within the interfering circle.
   However, for the ultra-dense network in which there are as many cells as UEs, each UE will be served by a certain small cell almost surely.
As a result, the interference to both CSG and OSG femto cells will be very small. In addition, the advantage of CSG femto cells over OSG femto cells will also be limited since their coverage is very limited.

The success probability of a CSG femto cell user, i.e., $\Pr\{\rho(\kappa)>T\}$, is presented in Fig. \ref{fig:psuc_cl} (a) and (b), where $T$ is a threshold, $\kappa=\frac{d_j}{R}$ is the ratio between the user's distance to the H-eNB and the cell radius.

\begin{figure}[h]
\centering
\includegraphics[width=3.0in]{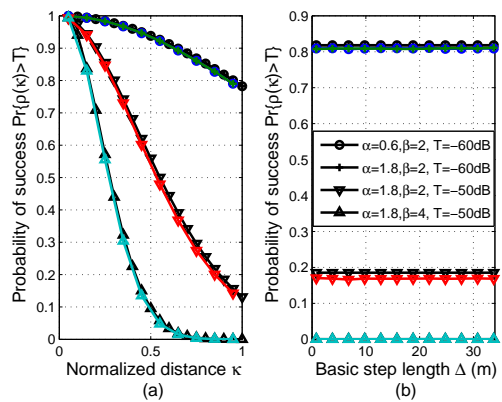}
\caption{The success probability of CSG femto cells. In (a), $\Delta=30$. In (b), $\kappa=0.9$. Black curves corresponds to theoretical results and colored curves corresponds to Monte Carlo results.}\label{fig:psuc_cl}
\end{figure}
It is seen from  Fig. \ref{fig:psuc_cl} (a) that the success probability will decrease if $\kappa$ is increased, which is due to the large scale signal attenuation.
   The L\'{e}vy flight parameter $\alpha$ for the curve marked by circle and the curve marked by `$+$' are $\alpha=0.6$ and $\alpha=1.8$, respectively.
By its physical meaning, larger $\alpha$ means that the  user will take more short flights and be less mobile.
As is shown, this does not change the success probability.
   But if we use a larger SINR threshold $T=-50$dB, the success probability becomes much smaller, as shown by the curve marked by `$\nabla$'.
Finally, the curve at the bottom corresponds to a lager pathloss exponent $\beta=4$. Due to the serious attenuation and interference, its success probability is the smallest.
Similar observations are obtained in Fig. \ref{fig:psuc_cl} (b), which presents the relationship between success probability and user velocity. It is seen that the basic step length will not change the success probability either.

\begin{figure}[h]
\centering
\includegraphics[width=3.0in]{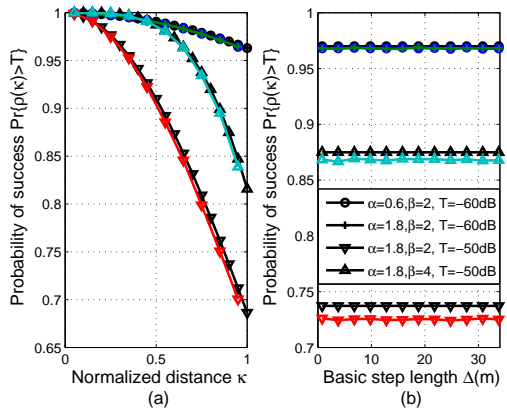}
\caption{The success probability of OSG femto cells. In (a), $\Delta=3$. In (b), $\kappa=0.9$. In (b), $\kappa=0.9$. Black curves corresponds to theoretical results and colored curves corresponds to Monte Carlo results.}\label{fig:psuc_op}
\end{figure}
\begin{figure}[h]
\centering
\includegraphics[width=3.0in]{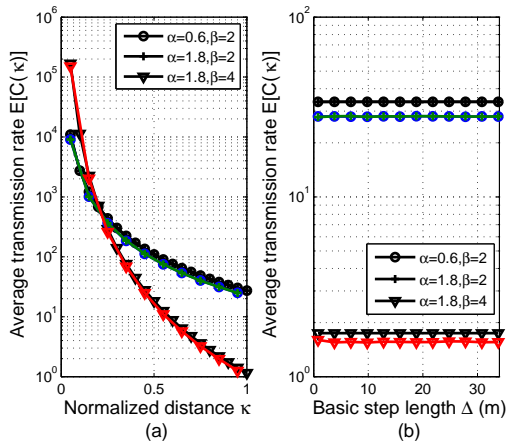}
\caption{The average rate of CSG femto cells (nat/s). In (a), $\Delta=3$. In (b), $\kappa=0.9$. In (b), $\kappa=0.9$. Black curves corresponds to theoretical results and colored curves corresponds to Monte Carlo results.}\label{fig:pec_cl}
\end{figure}
Fig. \ref{fig:psuc_op} presents the success probability of an OSG femto user. Likewise, success probability is smaller for larger $T$ and does not change with  $\alpha$ or $\Delta$.
   However, if the pathloss exponent $\beta$ is increased from 2 to 4, the success probability also becomes larger, which is contrary to the case of CSG femto cells.
In both Fig. \ref{fig:psuc_op} (a) and Fig. \ref{fig:psuc_op} (b), the curve corresponding to $\beta=4$ (labeled by `$\triangle$') achieves better performance than that corresponding to $\beta=2$ (labeled by `$\nabla$').
   In fact, although the desired signal has higher attenuation when $\beta$ is larger, the attenuation of the uplink interference to OSG femto users will be higher. As a result, the system performance gets better.

\begin{figure}[h]
\centering
\includegraphics[width=3.0in]{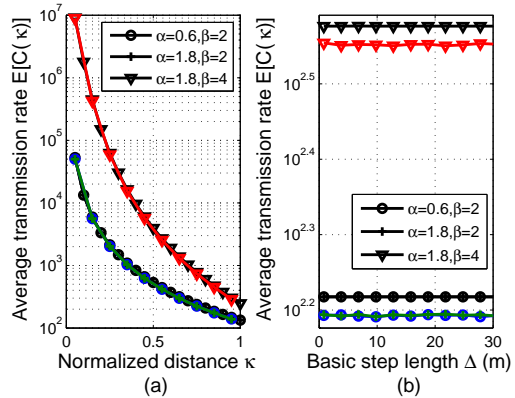}
\caption{The average rate of OSG femto cells (nat/s). In (a), $\Delta=3$. In (b), $\kappa=0.9$. In (b), $\kappa=0.9$. Black curves corresponds to theoretical results and colored curves corresponds to Monte Carlo results.}\label{fig:pec_op}
\end{figure}
Fig. \ref{fig:pec_cl} and Fig. \ref{fig:pec_op} present H-UE's average uplink transmission rate, for both CSG femto cells and OSG femto cells.
In both cases, L\'{e}vy flight parameter $\alpha$ and basic step length $\Delta$ will not affect the average rate.
 The figures also show that the performance of OSG femto cells will be better when $\beta$ is large, which is contrary to CSG femto cells.

\begin{figure}[h]
\centering
\includegraphics[width=3.0in]{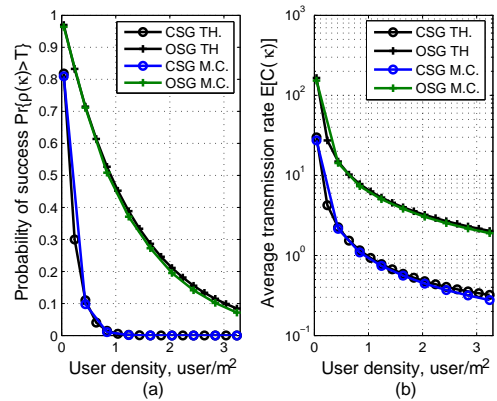}
\caption{Success probability and average rate vs user density. $\Delta=3,T=-60 dB,\kappa=0.9$. Theoretical results and Monte Carlo Results are labeled by `TH.' and `M.C.' respectively}\label{fig:pdens}
\end{figure}
Finally, the scaling of the performance of CSG and OSG femto cells with user densities, i.e., number of users per square meter, is evaluated in Fig. \ref{fig:pdens}. It is seen that, both success probability and average transmission rate  decrease with denser density, in which OSG femto cells perform better. Actually, the interference is larger when the network supports more users. The total throughput of the network, which can be obtained readily, is still increasing with the increase of user density.


\section{Conclusion}\label{sec:8}
In this work,  a mobility-aware uplink interference model for 5G heterogeneous networks was proposed.
   The proposed interference model provides an initial tool to evaluate the system performance in terms of success probability and average rate.
From our work, some interesting insights can be drawn to help the system designs.
   As is shown, the the statistics of the strength of uplink interference will not change with user velocity.
It is also seen that OSG femto users perform better since they suffer from less interference than CSG femto users.
   In addition, large pathloss exponents makes OSG femto uers perform even better but degrades the performance of CSG femto users.
Intuitively, using OSG femto cells can guarantee better frequency reuse and spectrum efficiency.  Therefore, it is suggested that more OSG femto cells or public pico/relay cells be deployed, other than CSG femto cells.

\section*{Acknowledgement}
This work was supported by the China Major State Basic Research Development Program (973 Program) No.2012CB316100(2), National Natural Science Foundation of China (NSFC) No.61171064 and NSFC No.61021001, and the Foundation from Tsinghua National Laboratory for Information Science and Technology.

\appendices
\renewcommand{\theequation}{\thesection.\arabic{equation}}

\newcounter{mytempthcnt}
\setcounter{mytempthcnt}{\value{theorem}}
\setcounter{theorem}{2}

\section{Proof of Lemma 1}\label{prf:lem1}
\begin{proof}
Assume that UE$_i$ locates at point $O$ at the beginning. Let $\boldsymbol{S}$ be arbitrary part of the macro cell with area $A_S$. Since $O$ is uniformly distributed in the macro cell, the probability that $O$ falls into $\boldsymbol{S}$ is $\Pr\{O\in \boldsymbol{S}\}=\frac{A_S}{\pi L^2}$, where $L$ is the radius of the macro cell.

Denote the flight of UE$_i$ as $(\rho, \theta)$ and the end point of the flight as $O'$. Then Lemma \ref{lem:unif_distri} will be proved if $\Pr\{O'\in \boldsymbol{S}\}=\frac{A_S}{\pi L^2}$ also holds, for arbitrary $\boldsymbol{S}$ in the macro cell.

By the modified reflection model, $O'$ will be in the macro cell, no matter how large $\rho$ is or what is $\theta$.
Therefore, the moving from point $O$ to point $O'$ can be seen as a linear mapping $\boldsymbol{F}_f(\cdot)$, from $\boldsymbol{\Omega}$ to $\boldsymbol{\Omega}$, where $\boldsymbol{\Omega}$ is closure of all points in the macro cell. Besides, it is also seen that the inverse mapping $\boldsymbol{F}_f^{-1}(\cdot)$ is also linear.
In this way, this flight can be interpreted as $O'=\boldsymbol{F}_f(O)$.
By using the inverse mapping $\boldsymbol{F}_f^{-1}(\cdot)$ to all the points within $\boldsymbol{S}$, a new region $\boldsymbol{S}_f^{-1}$ can be obtained, which also locates within the macro cell. Since the mapping is linear, it is seen that the area of $\boldsymbol{S}_f^{-1}$ is $A_S$.

By the definition of uniform distribution, we have $\Pr\{O\in \boldsymbol{S}^{-1}\}=\Pr\{O\in \boldsymbol{S}\}=\frac{A_S}{\pi L^2}$. Furthermore, the probability that $O'$ will be in $\boldsymbol{S}$ can be given by

\begin{equation}\label{dr:lem1}
\begin{split}
  \Pr\{O'\in \boldsymbol{S}\}=&\mathbb{E}_f [\Pr\{O'\in \boldsymbol{S}| \textrm{~the~flight~is~} (\rho, \theta)\}]\\
                =&\mathbb{E}_f [\Pr\{\boldsymbol{F}_f^{-1}(O') \in \boldsymbol{F}_f^{-1}(\boldsymbol{S})| \textrm{~the~flight~is~}(\rho, \theta)\} ]\\
                =&\mathbb{E}_f [ \Pr\{ O\in \boldsymbol{S}_f^{-1} | \textrm{~the~flight~is~}(\rho, \theta) \} ]\\
                =&\mathbb{E}_f [ \Pr\{ O\in \boldsymbol{S} | \textrm{~the~flight~is~}(\rho, \theta) \} ]\\
                =&\Pr\{ O\in \boldsymbol{S} \}=\frac{A_S}{\pi L^2}.
\end{split}
\end{equation}

Since $\boldsymbol{S}$ is arbitrary, we know that the location a UE$_i$ is uniformly distributed after one move. This also means that the location of UE$_i$ will be uniformly distributed throughout the operation, which completes the proof of Lemma \ref{lem:unif_distri}.
\end{proof}

{\small
\bibliographystyle{IEEEtran}

}

\end{document}